\documentclass[a4paper,10pt,twoside]{article}

\usepackage{dcolumn}%
\usepackage{a4wide}%
\usepackage{bm}%
\usepackage{amsmath,amssymb,amsfonts,amsthm, MnSymbol}
\usepackage{moreverb,rotating,graphics}
\usepackage[T1]{fontenc}
\usepackage[utf8]{inputenc}
\usepackage{epsfig}
\usepackage[francais,english]{babel} 
\usepackage{tikz}
\usepackage{mathrsfs}

\newtheorem{theorem}{Theorem}

\newtheorem{remark}{Remark}
\newtheorem{lemma}{Lemma}
\newtheorem{corollary}{Corollary}

\def\C{{\mathbb C}}

\def\N{{\mathbb N}}

\def\R{{\mathbb R}}

\def\bj{{\bold j}}

\def\bm{{\bold m}}

\def\br{{\bold r}}

\def\bv{{\bold v}}

\def\bw{{\bold w}}

\def\bx{{\bold x}}

\def\bz{{\bold z}}

\def\bsigma{{\bold \sigma}}

\def\bnull{{\bold 0}}

\def\rd{{\mathrm{d}}}
\def\re{{\mathrm{e}}}
\def\ri{{\mathrm{i}}}

\def\Re{{\mathrm{Re }\, }}
\def\Im{{\mathrm{Im }\, }}

\def\tr{{\rm tr}}

\def\cC{{\mathcal C}}

\def\cF{{\mathcal F}}
\def\cG{{\mathcal G}}

\def\cI{{\mathcal I}}
\def\cJ{{\mathcal J}}

\def\cM{{\mathcal M}}

\def\cS{{\mathcal S}}

\def\cW{{\mathcal W}}

\newcommand{\sS}{\mathscr{S}}

\newcommand{\bra}{\langle}
\newcommand{\ket}{\rangle}

\def\spinup{{\uparrow}}
\def\spindown{{\downarrow}}

\newcommand{\curl}{\textbf{curl}} 

\newcommand{\slater}{\mathrm{Slater}}
\newcommand{\mixed}{\mathrm{mixed}}
\newcommand{\pure}{\mathrm{pure}}

\begin{document}

\title{Pure-state $N$-representability in current-spin-density-functional theory}
\author{David Gontier~\\
Universit\'e Paris Est, CERMICS (ENPC), INRIA, F-77455 Marne-la-Vall\'ee}

\maketitle

\begin{abstract}
	This paper is concerned with the pure-state $N$-representability problem for systems under a magnetic field. Necessary and sufficient conditions are given for a spin-density $2 \times 2$ matrix $R$ to be representable by a Slater determinant. We also provide sufficient conditions on the paramagnetic current $\bj$ for the pair $(R, \bj)$ to be Slater-representable in the case where the number of electrons $N$ is greater than 12. The case $N < 12$ is left open.
\end{abstract}

\section{Introduction}

The density-functional theory (DFT), first developed by Hohenberg and Kohn~\cite{Hohenberg1964}, then further developed and formalized mathematically by Levy~\cite{Levy1979}, Valone~\cite{Valone1980} and Lieb~\cite{Lieb1983}, states that the ground state energy and density of a non-magnetic electronic system can be obtained by minimizing some functional of the density only, over the set of all admissible densities. Characterizing this set is called the \textit{N-representability problem}. More precisely, as the so-called constrained search method leading to DFT can be performed either with $N$-electron wave functions~\cite{Levy1979, Lieb1983}, or with $N$-body density matrices~\cite{Valone1980,Lieb1983}, the $N$-representability problem can be recast as follows: \textit{What is the set of electronic densities that come from an admissible $N$-electron wave function?} (pure-state $N$-representability) and \textit{What is the set of electronic densities that come from an admissible $N$-electron density matrix?} (mixed-state $N$-representability) This question was answered by Gilbert~\cite{Gilbert1975}, Harriman~\cite{Harriman1981} and Lieb~\cite{Lieb1983} (see also Remark~\ref{rem:rep_rho}).\\

For a system subjected to a magnetic field, the energy of the ground state can be obtained by a minimization over the set of pairs $(R,{\bold j})$, where $R$ denotes the $2 \times 2$ spin-density matrix~\cite{Gontier2013} (from which we recover the standard electronic density $\rho$ and the spin angular momentum density $\bold m$) and $\bold j$ the paramagnetic current~\cite{Vignale1988}. This has lead to several density-based theories, that come from several different approximations. In spin-density-functional theory (SDFT), one is only interested in spin effects, hence the paramagnetic term is neglected. The SDFT energy functional of the system therefore only depends on the spin-density $R$. The $N$-representability problem in SDFT are therefore: \textit{What is the set of spin-densities that come from an admissible $N$-electron wave function?} (pure-state representability) and \textit{What is the set of spin-densities that come from an admissible $N$-body density matrices?} (mixed-state representability). This question was left open in the pioneering work by von Barth and Hedin~\cite{Barth1972}, and was answered recently in the mixed-case setting~\cite{Gontier2013}. In parallel, in current-density-functional theory (CDFT), one is only interested in magnetic orbital effects, and spin effects are neglected~\cite{Vignale1987}. In this case, the CDFT energy functional of the system only depends on $\rho$ and $\bj$, and we need a characterization of the set of pure-state and mixed-state $N$-representable pairs $(\rho, \bj)$. Such a characterization was given recently by Hellgren, Kvaal and Helgaker in the mixed-state setting~\cite{Tellgren2014}, and by Lieb and Schrader in the pure-state setting, when the number of electrons is greater than 4~\cite{Lieb2013}. In the latter article, the authors rely on the so-called Lazarev-Lieb orthogonalization process~\cite{Lazarev2014} (see also Lemma~\ref{lem:HR}) in order to orthogonalize the Slater orbitals. \\

The purpose of this article is to give an answer to the $N$-representability problem in the current-spin-density-functional theory (CSDFT): \textit{What is the set of pairs $(R, \bj)$ that come from an admissible $N$-electron wave-function?} (pure-state) and  \textit{What is the set of pairs $(R, \bj)$ that come from an admissible $N$-body density-matrix?} (mixed-state). We will answer the question in the mixed-state setting for all $N \in \N^*$, and in the pure-state setting when $N \ge 12$ by combining the results in \cite{Gontier2013} and in \cite{Lieb2013}. In the process, we will answer the $N$-representability problem for SDFT for all $N \in \N^*$ in the pure-state setting. The proof relies on the Lazarev-Lieb orthogonalization process. In particular, our method does not give an upper-bound for the kinetic energy of the wave-function in terms of the previous quantities (we refer to \cite{Lazarev2014,Rutherfoord2013} for more details). We leave open the case $N < 12$ for pure-state CSDFT representability. \\

The article is organized as follows. In Section \ref{sec:states}, we recall briefly what are the sets of interest. We present our main results in Section~\ref{sec:results}, the proofs of which are given in Section~\ref{sec:proofs}.



\section{The different Slater-state, pure-state and mixed-state sets}
\label{sec:states}

We recall in this section the definition of Slater-states, pure-states and mixed-states. We denote by $L^p(\R^3)$, $H^1(\R^3)$, $C^\infty(\R^3)$, ... the spaces of \textit{real-valued} $L^p$, $H^1$, $C^\infty$, ... functions on $\R^3$, and by $L^p(\R^3, \C^d)$, $H^1(\R^3, \C^d)$, $C^\infty(\R^3, \C^d)$, ... the spaces of \textit{$\C^d$-valued} $L^p$, $H^1$, $C^\infty$ functions on $\R^3$. We will also make the identification $L^p(\R^3, \C^d)\equiv (L^p(\R^3, \C))^d$ (and so on). The one-electron state space is
\[
	L^2(\R^3, \C^2) \equiv \left\{ \Phi = (\phi^\spinup, \phi^\spindown)^T, \ \| \Phi \|_{L^2} := \int_{\R^3} | \phi^\spinup |^2 + | \phi^\spindown|^2 < \infty \right\},
\]
endowed with the natural scalar product $\bra \Phi_1 | \Phi_2 \ket := \int_{\R^3} \left( \overline{\phi_1^\spinup} \phi_2^\spinup + \overline{\phi_1^\spindown} \phi_2^\spindown  \right)$. The Hilbert space for $N$-electrons is the fermionic space $\bigwedge_{i=1}^N L^2(\R^3, \C^2)$ which is the set of wave-functions $\Psi \in L^2((\R^3,\C^2)^N)$ satisfying the Pauli-principle: for all permutations $p$ of $\{1, \ldots, N\}$,
\[
	\Psi(\br_{p(1)}, s_{p(1)}, \ldots, \br_{p(N)}, s_{p(N)}) = \varepsilon(p) \Psi(\br_1, s_1,  \ldots, \br_N, s_N),
\]
where $\varepsilon(p)$ denotes the parity of the permutation $p$, $\br_k \in \R^3$ the position of the $k$-th electron, and $s_k \in \{ \spinup, \spindown\}$ its spin. The set of admissible wave-functions, also called the set of pure-states, is the set of normalized wave-function with finite kinetic energy
\[
	\cW_N^\pure := \left\{ \Psi \in \bigwedge_{i=1}^N L^2(\R^3, \C^2), \| \nabla \Psi \|_{L^2}^2 < \infty, \| \Psi \|_{L^2(\R^{3N})}^2 = 1 \right\}
\]
where $\nabla$ is the gradient with respect to the $3N$ position variables. A special case of wave-functions is given by Slater determinants: let $\Phi_1, \Phi_2, \ldots, \Phi_N$ be a set of orthonormal functions in $L^2(\R^3, \C^2)$, the Slater determinant generated by $(\Phi_1, \ldots, \Phi_N)$ is (we denote by $\bx_k := (\br_k, s_k)$ the $k$-th spatial-spin component)
\[
	\sS \left[\Phi_1, \ldots, \Phi_N \right] (\bx_1, \ldots, \bx_N) := \dfrac{1}{\sqrt{N!}} \det \left( \Phi_i (\bx_j) \right)_{1 \le i,j \le N}.
\]
The subset of $\cW_N^\pure$ consisting of all finite energy Slater determinants is noted $\cW_N^\slater$. It holds that $\cW_1^\slater = \cW_1^\pure$, and, $\cW_N^\slater \subsetneq \cW_N^\pure$ for $N \ge 2$. \\


For a wave-function $\Psi \in \cW_N^\pure$, we define the corresponding $N$-body density matrix $\Gamma_\Psi = | \Psi \ket \bra \Psi |$, which corresponds to the projection on $\{ \C \Psi \}$ in $\bigwedge_{i=1}^N L^2(\R^3, \C^2)$. The set of pure-state (resp. Slater-state) $N$-body density matrices is
\[
	G_N^\pure := \left\{ \Gamma_\Psi,  \Psi \in \cW_N^\pure \right\} \ \text{resp.} \ G_N^\slater := \left\{ \Gamma_\Psi, \Psi \in \cW_N^\slater \right\}.
\]
It holds that $G_1^\slater = G_1^\pure$ and that $G_N^\slater \subsetneq G_N^\pure$ for $N \ge 2$. The set of mixed-state $N$-body density matrices $G_N^\mixed$ is defined as the convex hull of $G_N^\pure$:
\[
	G_N^\mixed = \left\{ \sum_{k=1}^\infty n_k | \Psi_k \ket \Psi_k | ,  0 \le n_k \le 1,  \sum_{k=1}^{\infty} n_k = 1,  \Psi_k \in \cW_N^\pure \right\}.
\]
It is also the convex hull of $G_N^\slater$. The kernel of an operator $\Gamma \in G_N^\mixed$ will be denoted by
\[
	\Gamma(\br_1, s_1, \ldots, \br_N, s_N; \br_1', s_1', \ldots, \br_N', s_N').
\]


The quantities of interest in density-functional theory are the spin-density $2 \times 2$ matrix, and the paramagnetic-current. For $\Gamma \in G_N^\mixed$, the associated spin-density $2 \times 2$ matrix is the $2 \times 2$ hermitian function-valued matrix
\[
	R_\Gamma (\br) := \begin{pmatrix} \rho^{\spinup \spinup}_\Gamma & \rho^{\spinup \spindown}_\Gamma \\  \rho^{\spindown \spinup}_\Gamma &  \rho^{\spindown \spindown}_\Gamma \end{pmatrix} (\br),
\]
where, for $\alpha,\beta \in \{ \spinup, \spindown \}^2$,
\[
	\rho^{\alpha \beta}_\Gamma (\br) := N \sum_{\vec{s} \in \{ \spinup, \spindown\}^{(N-1)}} \int_{\R^{3(N-1)}} 
	\Gamma(\br, \alpha, \vec{\bz}, \vec{s}; \br, \beta, \vec{\bz}, \vec{s}) \ \rd \vec{\bz}.
\]
In the case where $\Gamma$ comes from a Slater determinant $\sS[\Phi_1, \ldots, \Phi_N]$, we get
\begin{equation} \label{eq:DM_Slater}
	R_\Gamma(\br) = \sum_{k=1}^N \begin{pmatrix} | \phi^\spinup_k |^2 & \phi^\spinup_k \overline{\phi^\spindown_k} \\ \overline{\phi^\spinup_k} \phi^\spindown_k & | \phi^\spindown_k |^2 \end{pmatrix}.
\end{equation}
The total electronic density is $\rho_\Gamma = \rho_\Gamma^{\spinup \spinup} + \rho_\Gamma^{\spindown \spindown}$, and the spin angular momentum density is $\bm_\Gamma = \tr_{\C^2} [ \bsigma R_\Gamma]$, where
\[
	\bsigma := \left( \sigma_{x}, \sigma_{y}, \sigma_{z} \right) = \left( \begin{pmatrix} 0 & 1 \\ 1 & 0 \end{pmatrix} , \begin{pmatrix} 0 & -\ri \\ \ri & 0 \end{pmatrix}, \begin{pmatrix}  1 & 0 \\ 0 & -1 \end{pmatrix}     \right)
\]
contains the Pauli-matrices. Note that the pair $(\rho_\Gamma, \bm_\Gamma)$ contains the same information as $R_\Gamma$, hence the $N$-representability problem for the matrix $R$ is the same as the one for the pair $(\rho, \bm)$. However, as noticed in \cite{Gontier2013}, it is more natural mathematicaly speaking to work with $R_\Gamma$. The Slater-state, pure-state and mixed-state sets of spin-density $2 \times 2$ matrices are respectively defined by
\begin{align*}
	\cJ_N^\slater & := \left\{ R_\Gamma, \ \Gamma \in G_N^\slater \right\}, \\
	\cJ_N^\pure & := \left\{ R_\Gamma, \ \Gamma \in G_N^\pure \right\}, \\
	\cJ_N^\mixed & := \left\{ R_\Gamma, \ \Gamma \in G_N^\mixed \right\}.
\end{align*}
Since the map $\Gamma \mapsto R_\Gamma$ is linear, it holds that $\cJ_N^\slater \subset \cJ_N^\pure \subset \cJ_N^\mixed$, that $\cJ_N^\mixed$ is convex, and is the convex hull of both $\cJ_N^\slater$ and $\cJ_N^\pure$.\\

For a $N$-body density matrix $\Gamma \in G_N^\mixed$, we define the associated paramagnetic current $\bj_\Gamma  = \bj_\Gamma^\spinup + \bj_\Gamma^\spindown$ with
\[
	\bj_\Gamma^\alpha = \Im \left( N \hspace{-2ex} \sum_{\vec{s} \in \{ \spinup, \spindown\}^{N-1}} \int_{\R^{3(N-1)}}  \nabla_{\br'} \Gamma(\br, \alpha, \vec{\bz}, \vec{s} ; \br', \alpha, \vec{\bz}, \vec{s}) \Big|_{\br' = \br} \rd \vec{\bz} \right).
\]
In the case where $\Gamma$ comes from a Slater determinant $\cS[\Phi_1, \ldots, \Phi_N]$, we get
\begin{equation} \label{eq:j_Slater}
	\bj_{\Gamma} = \sum_{k=1}^N \Im \left( \overline{\phi_k^\spinup} \nabla \phi_k^\spinup + \overline{\phi_k^\spindown} \nabla \phi_k^\spindown \right).
\end{equation}
Note that while only the total paramagnetic current $\bj$ appears in the theory of C(S)DFT, the pair $(\bj^\spinup, \bj^\spindown)$ is sometimes used to design accurate current-density functionals (see \cite{Vignale1988} for instance). In this article however, we will only focus on the representability of $\bj$.\\


\section{Main results}
\label{sec:results}

\subsection{Representability in SDFT}

Our first result concerns the characterization of $\cJ_N^\slater$, $\cJ_N^\pure$ and $\cJ_N^\mixed$. For this purpose, we introduce
\begin{equation} \label{eq:CN}
\begin{aligned}
	\cC_N  := & \Big\{ R \in \cM_{2 \times 2} (L^1(\R^3, \C)), \ R^* = R, \quad R \ge 0,  \\ 
		&   \int_{\R^3} \tr_{\C^2} \left[ R \right] = N, \quad \sqrt{R} \in \cM_{2 \times 2} (H^1(\R^3, \C)) \Big\},
\end{aligned}
\end{equation}
and $\cC_N^0 := \left\{ R \in \cC_N, \ \det R \equiv 0 \right\}$. The following characterization of $\cC_N$ was proved in~\cite{Gontier2013}.
\begin{lemma} \label{lem:CN}
A function-valued matrix $R = \begin{pmatrix} \rho^\spinup & \sigma \\ \overline{\sigma} & \rho^\spindown \end{pmatrix}$ is in $\cC_N$ iff its coefficients satisfy
\begin{equation} \label{eq:conditions}
	\left\{ \begin{aligned}
	& \rho^{\spinup/\spindown} \ge 0,\quad \rho^\spinup \rho^\spindown - | \sigma |^2 \ge 0, \quad \int \rho^\spinup + \int \rho^\spindown = N, \\
	& \sqrt{\rho^{\spinup/ \spindown}} \in  H^1(\R^3), \quad \sigma, \sqrt{\det (R)} \in W^{1, 3/2}(\R^3), \\
	& | \nabla \sigma |^2 \rho^{-1}  \in L^1(\R^3), \\
	& \left| \nabla \sqrt{\det(R)} \right|^2 \rho^{-1} \in  L^1(\R^3) .
	\end{aligned} \right.
\end{equation}
\end{lemma}

The complete answer for $N$-representability in SDFT is given by the following theorem, whose proof is given in Section~\ref{sec:proof_SDFT}.

\begin{theorem} \label{th:SDFT} $\,$ \\
\noindent \textbf{Case $N=1$:} It holds that
\[
	\cJ_1^\slater = \cJ_1^\pure = \cC_1^0 \quad \text{and}  \quad \cJ_1^\mixed = \cC_1.
\]
\noindent \textbf{Case $N \ge 2$:} For all $N \ge 2$, it holds that
\[
	\cJ_N^\slater = \cJ_N^\pure = \cJ_N^\mixed = \cC_N.
\]

\end{theorem}

Note that the equality $\cJ_N^\mixed = \cC_N^\mixed$ for all $N \in \N^*$ was already proven in~\cite{Gontier2013}.

\begin{remark} \label{rem:rep_rho}
	Gilbert~\cite{Gilbert1975}, Harriman~\cite{Harriman1981} and Lieb~\cite{Lieb1983} proved that the $N$-representability set for the total electronic density $\rho$ is the same for Slater-states, pure-states and mixed-states, and is characterized by
	\begin{equation} \label{eq:Lieb1983}
		\cI_N := \left\{ \rho \in L^1(\R^3), \ \rho \ge 0, \ \int_{\R^3} \rho = N, \ \sqrt{\rho} \in H^1(\R^3) \right\}.
	\end{equation}
Comparing~\eqref{eq:Lieb1983} and~\eqref{eq:CN}, we see that our theorem is a natural extension of the previous result.
\end{remark}

\subsection{Representability in CSDFT}

We first recall some classical necessary conditions for a pair $(R, \bj)$ to be $N$-representable (we refer to~\cite{Tellgren2014, Lieb2013} for the proof). In the sequel, we will denote by $\rho^\spinup := \rho^{\spinup \spinup}$, $\rho^\spindown := \rho^{\spindown \spindown}$ and $\sigma := \rho^{\spinup \spindown}$ the elements of a matrix $R$, so that $R = \begin{pmatrix} \rho^\spinup & \sigma \\ \overline{\sigma} & \rho^\spindown \end{pmatrix}$, and by $\rho = \rho^\spinup + \rho^\spindown$ the associated total electronic density.
\begin{lemma}
	If a pair $(R, \bj)$ is representable by a mixed-state $N$-body density matrix, then
	\begin{equation} \label{eq:necessary_j}
		\left\{ \begin{array}{l}
			R \in \cC_N \\
			| \bj |^2 / \rho \in L^1(\R^3). \\
		\end{array} \right.
	\end{equation}
\end{lemma}

From the second condition of~\eqref{eq:necessary_j}, it must hold that the support of $\bj$ is contained in the support of $\rho$. The vector $\bv := \rho^{-1} \bj$ is called the velocity field, and $\bw := \curl(\bv)$ is called the vorticity.\\

Let us first consider the pure-state setting. Recall that in the spin-less setting, in the case $N=1$, a pair $(\rho, \bj)$ representable by a single orbital generally satisfies (provided that the phases of the orbital are globally well-defined) the curl-free condition $\curl (\rho^{-1} \bj ) = \bnull$ (see~\cite{Lieb2013, Tellgren2014}). This is no longer the case when spin is considered, as is shown is the following Lemma, whose proof is postponed until Section~\ref{sec:proof_CSDFT_N=1}.

\begin{lemma}[CSDFT, case $N=1$] \label{lem:CSDFT_N=1}
Let $\Phi = (\phi^\spinup, \phi^\spindown)^T \in \cW_1^\slater$ be such that both $\phi^\spinup$ and $\phi^\spindown$ have well-defined global phases in $C^1(\R)$. Then, the associated pair $(R, \bj)$ satisfies $R \in \cC_1^0$, ${| \bj |^2}/{\rho} \in L^1(\R^3)$, and the two curl-free conditions
\begin{equation}  \label{eq:cond_curlfree}
	\curl \left( \dfrac{\bj}{\rho} - \dfrac{\Im( \overline{\sigma} {\nabla \sigma})}{\rho \rho^\spindown} \right) = \bnull, \ 
	\curl \left( \dfrac{\bj}{\rho} + \dfrac{\Im( \overline{\sigma} {\nabla \sigma})}{\rho \rho^\spinup} \right)  = \bnull.
\end{equation}
\end{lemma}

\begin{remark}
	If we write $\sigma = | \sigma | \re^{\ri \tau}$, then, $| \sigma |^2 = \rho^{\spinup} \rho^{\spindown}$, and
	\begin{equation} \label{eq:nabla_tau}
		\Im \left( \overline{\sigma} { \nabla \sigma} \right) = | \sigma |^2 \nabla \tau = \rho^\spinup \rho^\spindown \nabla \tau.
	\end{equation}
	In particular, it holds that
	\[
		\curl \left( \dfrac{\Im( \overline{\sigma} {\nabla \sigma})}{\rho \rho^\spindown} + \dfrac{\Im( \overline{\sigma} {\nabla \sigma})}{\rho \rho^\spinup} \right) = \curl \ (\nabla \tau) = \bnull,
	\]
	so that one of the equalities in~\eqref{eq:cond_curlfree} implies the other one.
\end{remark}

\begin{remark}
	We recover the traditional result in the spin-less case, where $\sigma \equiv 0$.
\end{remark}


In the case $N > 1$, things are very different. In~\cite{Lieb2013}, the authors gave a rigorous proof for the representability of the pair $(\rho, \bj)$ by a Slater determinant (of orbitals having well-defined global phases) whenever $N \ge 4$ under a mild condition (see equation~\eqref{eq:Liebcond} below). By adapting their proof to our case, we are able to ensure representability of a pair $(R, \bj)$ by a Slater determinant for $N \ge 12$ under the same mild condition (see Section~\ref{sec:proof_CSDFT} for the proof).

\begin{theorem}[CSDFT, case $N \ge 12$] \label{th:CSDFT} $\,$ \\
A sufficient set of conditions for a pair $(R, \bj)$ to be representable by a Slater determinant is
\begin{itemize}
	\item $R \in \cC_N$ with $N \ge 12$ and $\bj$ satisfies $ | \bj |^2 / \rho \in L^1(\R^3)$
	\item there exists $\delta > 0$ such that,
\begin{equation} \label{eq:Liebcond}
		\sup_{\br \in \R^3} \ f(\br)^{(1+\delta)/2} | \bw (\br) | < \infty, \
		\sup_{\br \in \R^3} \ f(\br)^{(1+\delta)/2} | \nabla \bw (\br) | < \infty,
\end{equation}
where $\bw := \curl \ (\rho^{-2} \bj)$ is the vorticity, and
\[
	f(\br) := ( 1 + (r_1)^2)  ( 1 + (r_2)^2)  ( 1 + (r_3)^2).
\]
\end{itemize}

\end{theorem}


\begin{remark}
	The conditions~\eqref{eq:Liebcond} are the ones found in~\cite{Lieb2013}. The authors conjectured that this condition "can be considerably loosened".
\end{remark}

\begin{remark}
	We were only able to prove this theorem for $N \ge 12$. In~\cite{Lieb2013}, the authors proved that conditions~\eqref{eq:Liebcond} were not sufficient for $N=2$. We do not know whether conditions~\eqref{eq:Liebcond} are sufficient in the case $3 \le N \le 11$.
\end{remark}

Let us finally turn to the mixed-state case. We notice that if $(R, \bj)$ is representable by a Slater determinant $\sS[\Phi_1, \ldots, \Phi_N]$, then, for all $k \in \N^*$, the pair $(k/N) (R, \bj)$ is mixed-state representable, where $N$ is the number of orbitals (simply take the uniform convex combination of the pairs represented by $\sS[\Phi_1]$, $\sS[\Phi_2]$, etc.). In particular, from Theorem~\ref{th:CSDFT}, we deduce the following corollary.
\begin{corollary}[CSDFT, case mixed-state] $\,$ \\
A sufficient set of conditions for a pair $(R, \bj)$ to be mixed-state representable is $R \in \cC_N^0$ for some $N \in \N^*$, $\bj$ satisfies $ | \bj |^2 / \rho \in L^1(\R^3)$, and~\eqref{eq:Liebcond} holds for some $\delta>0$.
\end{corollary}

In~\cite{Tellgren2014}, the authors provide different sufficient conditions than~\eqref{eq:Liebcond} for a pair $(\rho, \bj)$ to be mixed-state representable, where $\rho$ is the electronic density. They proved that if
	\[
		(1 + | \cdot |^2) \rho \left| \nabla (\rho^{-1} \bj ) \right|^2 \in L^1(\R^3),
	\]
	then the pair $(\rho, \bj)$ is mixed-state representable. Their proof can be straightforwardly adapted for the representability of the pair $(R, \bj)$, so that similar results hold. The details are omitted here for the sake of brevity.



\section{Proofs}
\label{sec:proofs}

\subsection{Proof of Theorem~\ref{th:SDFT}}
\label{sec:proof_SDFT}

The mixed-state case was already proved in~\cite{Gontier2013}. We focus on the pure-state representability. \\

\noindent \textbf{Case $N=1$}\\
The fact that $\cJ_1^\slater = \cJ_1^\pure$ simply comes from the fact that $G_1^\slater = G_1^\pure$. To prove $\cJ_1^\slater \subset \cC_1^0$, we let $R \in \cJ_1^\slater$ be represented by $\Phi = (\phi^\spinup, \phi^\spindown)^T \in H^1(\R^3, \C^2)$, so that
\[
	R = \begin{pmatrix}  | \phi^\spinup |^2 & \phi^\spinup \overline{\phi^\spindown} \\
		\phi^\spindown \overline{\phi^\spinup} & | \phi^\spindown |^2
		\end{pmatrix}.
\]
Since $R \in \cJ_1^\slater \subset \cJ_1^\mixed = \cC_1$ and $\det R \equiv 0$, we deduce $R \in \cC_1^0$. \\
We now prove that $\cC_1^0 \subset \cJ_1^\slater$. Let $R = \begin{pmatrix} \rho^\spinup & \sigma \\ \overline{\sigma} & \rho^\spindown \end{pmatrix} \in \cC_1^0$. From $\det R \equiv 0$ and Lemma~\ref{lem:CN}, we get
\begin{equation} \label{eq:conditions0}
	\left\{ \begin{aligned}
	& \rho^{\spinup/\spindown} \ge 0,\quad \rho^\spinup \rho^\spindown = | \sigma |^2, \quad \int_{\R^3} \rho^\spinup + \int_{\R^3} \rho^\spindown = 1, \\
	& \sqrt{\rho^{\spinup/ \spindown}} \in  H^1(\R^3), \quad \sigma \in W^{1, 3/2}(\R^3), \\
	& | \nabla \sigma |^2 / \rho  \in L^1(\R^3). \\
	\end{aligned} \right.
\end{equation}
There are two natural choices that we would like to make for a representing orbital, namely
\begin{equation} \label{eq:Phi12}
	\Phi_1 = \begin{pmatrix} \sqrt{\rho^\spinup}, & \dfrac{\overline{\sigma}}{\sqrt{\rho^{\spinup}}} \end{pmatrix}^T
	\quad \text{and} \quad
	\Phi_2 = \begin{pmatrix} \dfrac{\sigma}{\sqrt{\rho^\spindown}}, & \sqrt{\rho^\spindown} \end{pmatrix}^T.
\end{equation}
Unfortunately, it is not guaranteed that these orbitals are indeed in $H^1(\R^3, \C^2)$. It is the case only if $| \nabla \sigma |^2 / \rho^{\spindown}$ is in $L^1(\R^3)$ for $\Phi_1$, and if $| \nabla \sigma |^2 / \rho^{\spinup}$ is in $L^1(\R^3)$ for $\Phi_2$. Due to~\eqref{eq:conditions0}, we know that $| \nabla \sigma |^2 / \rho \in L^1(\R^3)$. The idea is therefore to interpolate between these two orbitals, taking $\Phi_1$ in regions where $\rho^\spinup >> \rho^\spindown$, and $\Phi_2$ in regions where $\rho^{\spindown} >> \rho^{\spinup}$. This is done via the following process. \\

Let $\chi \in C^\infty(\R)$ be a non-decreasing function such that $0 \le \chi \le 1$, $\chi(x) = 0$ if $x \le 1/2$ and $\chi(x) = 1$ if $x \ge 1$. We write $\sigma = \alpha + \ri \beta$ where $\alpha$ is the real-part of $\sigma$, and $\beta$ is its imaginary part. We introduce
\[
	\begin{array}{llll}
	\lambda_1 & :=  \dfrac{\sqrt{ \alpha^2 + \chi^2(\rho^\spinup / \rho^\spindown) \beta^2}}{\sqrt{\rho^\spindown}}, & 
	\mu_1  & := \sqrt{1 - \chi^2( \rho^\spinup / \rho^\spindown)} \dfrac{\beta}{\sqrt{\rho^\spindown}}, \\
	\lambda_2 & := \dfrac{\alpha \lambda_1 + \beta \mu_1}{\rho^\spinup},  
	& \mu_2  & := \dfrac{\beta \lambda_1 - \alpha \mu_1}{\rho^\spinup},
	\end{array}
\]
and we set
\[
	\phi^\spinup := \lambda_1 + \ri \mu_1 \quad \text{and} \quad \phi^\spindown := \lambda_2 + \ri \mu_2.
\]
Let us prove that $\Phi$ represents $R$ and that $\Phi := (\phi^\spinup, \phi^\spindown) \in \cW_1^\slater$. First, an easy calculation shows that
\begin{align*}
	| \phi^\spinup |^2 & = \lambda_1^2 + \mu_1^2 = \dfrac{\alpha^2 + \chi^2 \beta^2 + (1 - \chi^2) \beta^2}{\rho^\spindown} = \dfrac{| \sigma |^2}{\rho^\spindown} = \rho^\spinup, \\
	| \phi^\spindown |^2 &  = \dfrac{(\alpha^2 + \beta^2)(\lambda_1^2 + \mu_1^2)}{(\rho^\spinup)^2} = \dfrac{| \sigma |^2}{\rho^{\spinup}} = \rho^\spindown, \\
	\Re \left( \phi^\spinup \overline{\phi^\spindown} \right) & = \lambda_1 \lambda_2 - \mu_1 \mu_2 = \dfrac{\alpha (\lambda_1^2  + \mu_1^2)}{\rho^\spinup} = \alpha, \\
	\Im \left( \phi^\spinup \overline{\phi^\spindown}  \right) & = \lambda_1 \mu_2 + \lambda_2 \mu_1 = \dfrac{\beta(\lambda_1^2 + \mu_1^2)}{\sqrt{\rho^\spinup}} = \beta,
\end{align*}
so that $\Phi \in L^2(\R^3, \C^2)$ with $\| \Phi \| = 1$, and $\Phi$ represents $R$. To prove that $\Phi \in \cW_1^\slater$, we need to check that $\lambda_1, \lambda_2, \mu_1$ and $\mu_2$ are in $H^1(\R^3)$. For $\lambda_1$, we choose another non-increasing function $\xi \in C^\infty(\R)$ such that $0 \le \xi \le 1$, $\xi(x) = 0$ for $x \le 1$, and $\xi(x) = 1$ for $x \ge 2$. Note that $(1 - \chi)\xi \equiv 0$. It holds that 
\begin{equation} \label{eq:nablaLambda}
	\nabla \lambda_1 = (1 - \xi^2(\rho^\spinup / \rho^\spindown)) \nabla \lambda_1 +  \xi^2(\rho^\spinup / \rho^\spindown) \nabla \lambda_1.
\end{equation}
The second term in the right-hand side of~\eqref{eq:nablaLambda} is non-null only if $\rho^\spinup \ge \rho^\spindown$, so that $\chi(\rho^\spinup / \rho^\spindown) = 1$ on this part. In particular, from the equality $\rho^\spinup \rho^\spindown = | \sigma |^2$, we get
\[
	 \xi^2(\rho^\spinup / \rho^\spindown) \lambda_1 =  \xi^2(\rho^\spinup / \rho^\spindown)  \dfrac{| \sigma |}{\sqrt{\rho^\spindown}} =  \xi^2(\rho^\spinup / \rho^\spindown) \sqrt{\rho^\spinup},
\]
and similarly,
\[
	 \xi^2(\rho^\spinup / \rho^\spindown) \nabla \lambda_1  = \xi^2(\rho^\spinup / \rho^\spindown) \nabla \sqrt{\rho^\spinup},
\]
which is in $L^2(\R^3)$ according to~\eqref{eq:conditions0}. On the other hand, the first term in the right-hand side of~\eqref{eq:nablaLambda} is non-null only if $\rho^\spinup \le 2 \rho^\spindown$, so that $(1/3) \rho \le \rho^\spindown$ on this part. In particular, from the following point-wise estimate
\[
	| \nabla \sqrt{f + g} | \le | \nabla \sqrt{f} | + | \nabla \sqrt{g} | ,
\]
which is valid almost everywhere whenever $f, g \ge 0$, the inequality $(a + b)^2 \le 2(a^2 + b^2)$, and the fact that $\alpha^2 + \chi^2 \beta^2 \le | \sigma |^2$, we get on this part (we write $\chi$ for $\chi(\rho^\spinup / \rho^\spindown)$)
\begin{align*}
	 \left| \nabla \lambda_1 \right|^2 
	& = \left| \dfrac{\sqrt{\rho^\spindown} \nabla \sqrt{  \alpha^2 + \chi^2 \beta^2  } - \sqrt{  \alpha^2 + \chi^2\beta^2 } \nabla \sqrt{\rho^\spindown}}{\rho^\spindown} \right|^2 \\
	& \le
	 2 \left( \dfrac{ | \nabla \sqrt{\alpha^2 +  \chi^2 \beta^2} |^2 } {\rho^\spindown} 
  + \dfrac{(\alpha^2 + \chi^2 \beta^2)}{(\rho^\spindown)^2} { | \nabla \sqrt{\rho^\spindown} |^2} \right) \\
	& \le
	2 \left( \dfrac{| \nabla \alpha |^2}{\rho^\spindown} 
	+ \dfrac{2 \left| \nabla \chi \frac{\rho^\spindown \nabla \rho^\spinup - \rho^\spinup \nabla \rho^\spindown}{(\rho^\spindown)^2} \right|^2 \beta^2   }{\rho^\spindown} \right. + \\
	&\qquad  \left. + \dfrac{2 \chi^2 | \nabla \beta |^2}{\rho^\spindown}
	+ \dfrac{2 | \sigma |^2}{(\rho^\spindown)^2}  | \nabla \sqrt{\rho^\spindown} |^2
	\right).
\end{align*}
We finally use the inequality $(\rho^\spindown )^{-1} \le (3/\rho)$, and the inequality $| \sigma |^2 / (\rho^{\spindown})^2 = \rho^{\spinup} / \rho^\spindown \le 2$ and get
\begin{align*}
	\left| \nabla \lambda_1 \right|^2 \le
	C & \left(  \dfrac{| \nabla \alpha |^2}{\rho} +  \| \nabla \chi \|_{L^\infty}^2 \left( \dfrac{| \nabla \rho^\spinup |^2}{\rho^\spinup} +  \dfrac{| \nabla \rho^\spindown |^2}{\rho^\spindown} \right) \right. \\
		& \qquad \left. + \dfrac{| \nabla \beta |^2}{\rho} + | \nabla \sqrt{\rho^\spindown} |^2
	\right).
\end{align*}

The right-hand side is in $L^1(\R^3)$ according to~\eqref{eq:conditions0}. Hence, $(1 - \xi^2(\rho^\spinup / \rho^\spindown)) \left| \nabla \lambda_1 \right| \in L^2(\R^3)$, and finally $\lambda_1 \in H^1(\R^3)$. \\
The other cases are treated similarly, observing that, 
\begin{itemize}
\item whenever $\rho^\spinup \ge \rho^\spindown$, then $\chi = 1$, and $\Phi = \Phi_1$ where $\Phi_1$ was defined in~\eqref{eq:Phi12}. We then control $(\rho^\spinup)^{-1}$ with the inequality $(\rho^\spinup)^{-1} \le 2 \rho^{-1}$ ;

\item whenever $\rho^\spinup \le  \rho^\spindown /2$, then  $\chi = 0$, $\Phi = \Phi_2$. We control $(\rho^\spindown)^{-1}$ with the inequality $(\rho^\spindown)^{-1} \le \frac32 \rho^{-1}$ ;

\item whenever $\rho^\spindown /2 \le \rho^\spinup \le \rho^\spindown$, then both $(\rho^\spinup)^{-1}$ and $\rho^{\spindown}$ are controlled via $(\rho^\spinup)^{-1} \le 3 \rho^{-1}$ and $(\rho^{\spindown})^{-1} \le 2 \rho^{-1} $.
\end{itemize}

\noindent \textbf{Case $N \ge 2$.} \\
Since $\cJ_N^\slater \subset \cJ_N^\pure \subset \cJ_N^\mixed = \cC_N$, it is enough to prove that $\cC_N \subset \cJ_N^\slater$. We start with a key lemma.

\begin{lemma} \label{lem:N+M}
	For all $M, N \in \N^2$, it holds that $\cJ_{N+M}^\slater = \cJ_N^\slater + \cJ_M^\slater$.
\end{lemma}

\begin{proof} [Proof of Lemma~\ref{lem:N+M}]
	The case $\cJ_{N+M}^\slater \subset \cJ_N^\slater + \cJ_M^\slater$ is trivial: if $R \in \cJ_{N+M}^\slater$ is represented by the Slater determinant $\sS[\Phi_1, \ldots \Phi_{N+M}]$, then, by denoting by $R_1$ (resp. $R_2$) the spin-density matrix associated to the Slater determinant $\sS[\Phi_1, \ldots, \Phi_N]$ (resp. $\sS[\Phi_{N+1}, \ldots, \Phi_{N+1}]$), it holds $R = R_1 + R_2$ (see Equation~\eqref{eq:DM_Slater} for instance), with $R_1 \in \cJ_N^\slater$ and $R_2 \in \cJ_M^\slater$. \\
	
	The converse is more involving, and requires an orthogonalization step. Let $R_1 \in \cJ_N^\slater$ be represented by the Slater determinant $\sS[\Phi_1, \ldots, \Phi_N]$, and $R_2 \in \cJ_M^\slater$ be represented by the Slater determinant $\sS[\widetilde \Phi_1, \ldots, \widetilde \Phi_M]$. We cannot directly consider the Slater determinant $\sS[\Phi_1, \ldots, \Phi_N, \tilde \Phi_1, \ldots, \tilde \Phi_M]$, for $(\Phi_1, \ldots, \Phi_N)$ is not orthogonal to $(\tilde \Phi_1, \ldots, \tilde \Phi_M)$. \\
	
	We recall the following lemma, which is a smooth version of the Hobby-Rice theorem \cite{Hobby1965} (see also \cite{Pinkus1976}), and that was proved by Lazarev and Lieb in \cite{Lazarev2014} (see also \cite{Lieb2013}).
	\begin{lemma}  \label{lem:HR}
	For all $N \in \N^*$, and for all $(f_1, \ldots, f_N) \in L^1(\R^3, \C)$, there exists a function $u \in C^\infty(\R^3)$, with bounded derivatives, such that
	\[
		\forall 1 \le k \le N, \quad \int_{\R^3} f_k \re^{\ri u} = 0.
	\]
	Moreover, $u$ can be chosen to vary in the $r_1$ direction only.
\end{lemma}
	We now modify the phases of $\widetilde{\Phi_1}, \ldots, \widetilde{\Phi_M} $ as follows. First, we choose $\widetilde{u_1}$ as in Lemma~\ref{lem:HR} such that, 
	\[
		\forall 1 \le k \le N, \quad \int_{\R^3} \left( \overline{\phi_k^\spinup} \widetilde{\phi_1^\spinup} +  \overline{\phi_k^\spindown} \widetilde{\phi_1^\spindown} \right) \re^{\ri \widetilde{u_1}} = 0,
	\]
	and we set $\Phi_{N+1} = \widetilde{\Phi_1} \re^{\ri \widetilde{u_1}}$. Note that, by construction, $\Phi_{N+1}$ is normalized, in $H^1(\R^3, \C^2)$, and orthogonal to $(\Phi_1, \ldots, \Phi_N)$. We then construct $\widetilde{u_2}$ as in Lemma~\ref{lem:HR} such that
	\[
		\forall 1 \le k \le N+1, \quad \int_{\R^3} \left( \overline{\phi_k^\spinup} \widetilde{\phi_2^\spinup} +  \overline{\phi_k^\spindown} \widetilde{\phi_2^\spindown} \right) \re^{\ri \widetilde{u_2}} = 0,
	\]
	and we set $\Phi_{N+2} = \widetilde{\Phi_2} \re^{\ri \widetilde{u_2}}$. We continue this process for $3 \le k \le M$ and construct $\Phi_{N+k} = \widetilde{\Phi_k} \re^{\ri \widetilde{u_k}}$. We thus obtain an orthonormal family $(\Phi_1, \ldots, \Phi_{N+M})$. By noticing that the spin-density matrix of the Slater determinant $\sS[ \widetilde{\Phi_1}, \ldots, \widetilde{ \Phi_M}]$ is the same as the one of $\sS[\Phi_{N+1}, \ldots, \Phi_{N+M}]$ (the phases cancel out), we obtain that $R = R_1 + R_2$, where $R$ is the spin-density matrix represented by $\sS[\Phi_{1}, \ldots, \Phi_{N+M}]$. The result follows.
	
\end{proof}

We now prove that $\cC_N \subset \cJ_N^\slater$. We start with the case $N=2$.\\

\noindent \textbf{Case $N = 2$.}\\
Let $R = \begin{pmatrix} \rho^\spinup & \sigma \\ \overline{\sigma} & \rho^\spindown \end{pmatrix} \in \cC_2$. We write $\sqrt{R} =  \begin{pmatrix} r^\spinup & s \\ \overline{s} & r^\spindown \end{pmatrix}$, with $r^{\spinup}, r^{\spindown} \in \left( H^1(\R^3) \right)^2$ and $s$ in $H^1(\R^3, \C)$. Let
\begin{equation} \label{eq:decomposition_convex}
	R^\spinup := \begin{pmatrix} | r^{\spinup} |^2   & s r^\spinup  \\ \overline{s} r^\spinup & | s |^2 \end{pmatrix} 
	\quad \text{and} \quad
	R^\spindown := \begin{pmatrix} | s |^2 & s r^\spindown  \\ \overline{s} r^\spindown & | r^\spindown |^2 \end{pmatrix}.
\end{equation}
It is easy to check $R = R^\spinup + R^\spindown$, that $R^{\spinup / \spindown}$ are hermitian, of null determinant, and $\sqrt{R^{\spinup / \spindown}} \in \cM_{2 \times 2} \left( H^1(\R^3, \C) \right)$. However, it may hold that $\int_{\R^3} \tr_{\C^2} [R^\spinup] \notin \N^*$, so that $R^\spinup$ is not in $\cC_M^0$ for some $M \in \N^*$. \\

The case $R^\spinup = 0$ or $R^\spindown = 0$ are trivial. Let us suppose that, for $\alpha \in \{\spinup, \spindown\}$, $m^\alpha := \int_{\R^3} \rho_{R^\alpha} \neq 0$. In this case, the matrices $\widetilde{R^\alpha} = (m^\alpha)^{-1} R^\alpha$ are in $\cC_1^0$, hence are representable by a single orbital, due to the first statement of Theorem~\ref{th:SDFT}. Let $\widetilde{\Phi} = (\widetilde{\phi_1^\spinup}, \widetilde{\phi_1^\spindown})^T \in H^1(\R^3, \C^2)$ and $\widetilde{\Phi}_2 = (\widetilde{\phi_2^\spinup}, \widetilde{\phi_2^\spindown})^T\in H^1(\R^3, \C^2)$ be normalized orbitals that represent respectively $\widetilde{R^\spinup}$ and $\widetilde{R^\spindown}$. It holds
\[
	\widetilde{\Phi}_1 \widetilde{\Phi}_1^* = \widetilde{R^\spinup} = (m^\spinup)^{-1} R^\spinup
	\quad \text{and} \quad 
	\widetilde{\Phi}_2 \widetilde{\Phi}_2^* =  \widetilde{R^\spindown}  = (m^\spindown)^{-1} R^\spindown.
\]
From the Lazarev-Lieb orthogonalization process (see Lemma~\ref{lem:HR}), there exists a function $u \in C^\infty(\R)$ with bounded derivatives such that
\begin{equation} \label{eq:for_norm}
	\bra \widetilde{\Phi}_1 | \widetilde{\Phi}_2 \re^{\ri u} \ket = \int_{\R^3} \left( \overline{\psi_1^\spinup} {\psi_2^\spinup} + \overline{\psi_1^\spindown} {\psi_2^\spindown} \right) \re^{ \ri u} = 0.
\end{equation}
Once this function is chosen, there exists a function $v \in C^\infty(\R)$ with bounded derivatives such that
\begin{equation} \label{eq:orthogonality12}
	\bra \widetilde{\Phi}_1 | \widetilde{\Phi}_1 \re^{\ri v} \ket = 
	\bra \widetilde{\Phi}_1 | \widetilde{\Phi}_2 \re^{\ri (u+v)} \ket = 
	\bra \widetilde{\Phi}_2 \re^{\ri u} | \widetilde{\Phi}_1 \re^{\ri v} \ket = 
	\bra \widetilde{\Phi}_2 | \widetilde{\Phi}_2 \re^{\ri v} \ket = 0.
\end{equation}
We finally set
\[
	\Phi_1 := \dfrac{1}{\sqrt{2}} \left( \sqrt{m^\spinup} \widetilde{\Phi}_1 + \sqrt{m^\spindown} \widetilde{\Phi}_2 \re^{\ri u}\right)
\]
and
\[
	 \Phi_2 :=  \dfrac{1}{\sqrt{2}} \left( \sqrt{m^\spinup} \widetilde{\Phi}_1 - \sqrt{m^\spindown} \widetilde{\Phi}_2 \re^{\ri u} \right) \re^{\ri v}.
\]
From~\eqref{eq:for_norm}, we deduce $\| \Phi_1 \|^2 = \| \Phi_2 \|^2 = 1$, so that both $\Phi_1$ and $\Phi_2$ are normalized. Also, from~\eqref{eq:orthogonality12}, we get $\bra \Phi_1 | \Phi_2 \ket = 0$, hence $\{ \Phi_1, \Phi_2\}$ is orthonormal. As $\widetilde{\Phi}_1$ and $\widetilde{\Phi}_2$ are in $H^1(\R^3, \C^2)$, and $u$ and $v$ have bounded derivatives, $\Phi_1$ and $\Phi_2$ are in $H^1(\R^3, \C^2)$. Finally, it holds that
\begin{align*}
	& \Phi_1 \Phi_1^* + \Phi_2 \Phi_2^* =\\
	& \quad = \dfrac12 \left( m^\spinup \widetilde{\Phi}_1 \widetilde{\Phi}_1^* + m^\spindown \widetilde{\Phi}_2 \widetilde{\Phi}_2^* + 2\sqrt{m^\spinup m^\spindown} \Re \left( \widetilde{\Phi}_1 \widetilde{\Phi}_2^* \re^{-\ri u} \right) \right. \\
	& \qquad + \left. m^\spinup \widetilde{\Phi}_1 \widetilde{\Phi}_1^* + m^\spindown \widetilde{\Phi}_2 \widetilde{\Phi}_2^* - 2\sqrt{m^\spinup m^\spindown} \Re \left( \widetilde{\Phi}_1 \widetilde{\Phi}_2^* \re^{-\ri u} \right) \right) \\
	& \quad = m^\spinup \widetilde{\Phi}_1 \widetilde{\Phi}_1^* + m^\spindown \widetilde{\Phi}_2 \widetilde{\Phi}_2^* = R.
\end{align*}
We deduce that the Slater determinant $\sS [ \Phi_1, \Phi_2 ]$ represents $R$, so that $R \in \cJ_2^\slater$. Altogether, $\cC_2 \subset \cJ_2^\slater$, and therefore $\cC_2 = \cJ_2^\slater$. \\

\noindent \textbf{Case $N > 2$.}\\
We proceed by induction. Let $R \in \cC_{N+1}$ with $N \ge 2$, and suppose $\cC_N = \cJ_N^\slater$. We use again the decomposition~\eqref{eq:decomposition_convex} and write $R = R^\spinup + R^\spindown$, where $R^{\spinup / \spindown}$ are two null-determinant hermitian matrices. For $\alpha  \in \{ \spinup, \spindown\}$, we note $m^\alpha := \int_{\R^3} \rho_{R^\alpha}$. Since $m^\spinup + m^\spindown = N+1 \ge 3$, at least $m^\spinup$ or $m^\spindown$ is greater than $1$. Let us suppose without loss of generality that $m^\spinup \ge 1$. We then write $R = R_1 + R_2 $ with
\[
	R_1 := (m^\spinup)^{-1} R^\spinup
	\quad \text{and} \quad
	R_2 :=  \left( \left( 1 - (m^\spinup)^{-1} \right) R^\spinup + m^\spindown R^\spindown \right).
\]
It holds that $R_1 \in \cC_1^0 = \cJ_1^\slater$ and $R_2 \in \cC_N = \cJ_N^\slater$ (by induction). Together with Lemma~\ref{lem:N+M}, we deduce that $R \in \cJ_{N+1}^\slater$. The result follows.


\subsection{Proof of Lemma~\ref{lem:CSDFT_N=1}}
\label{sec:proof_CSDFT_N=1}

Let $\Phi = (\phi^\spinup, \phi^\spindown) \in H^1(\R^3, \C^2)$ having well-defined global phases in $C^1(\R)$, and let $(R, \bj)$ be the associated spin-density matrix and paramagnetic current. It holds
\[
	R = \begin{pmatrix} \rho^\spinup & \sigma \\ \overline{\sigma} & \rho^\spindown \end{pmatrix} = 
	\begin{pmatrix} | \phi^\spinup |^2 & \phi^\spinup \overline{\phi^\spindown} \\
		\phi^\spindown \overline{\phi^\spinup} & | \phi^\spindown |^2 \end{pmatrix}.
\]
For $\alpha \in \{ \spinup, \spindown\}$, we let $\tau^\alpha$ be the phase of $\phi^\alpha$, so that $\phi^\alpha = \sqrt{\rho^\alpha} \re^{\ri \tau^\alpha}$. Setting $\tau = \tau^\spinup - \tau^\spindown$, we obtain $\sigma = | \sigma | \re^{\ri \tau} = \sqrt{\rho^\spinup \rho^\spindown} \re^{\ri \tau}$. The paramagnetic current is
\[
	\bj = \rho^\spinup \nabla \tau^\spinup + \rho^\spindown \nabla \tau^\spindown =  \rho \nabla \tau^\spindown + \rho^\spinup \nabla \tau =  \rho \nabla \tau^\spinup - \rho^\spindown \nabla \tau.
\]
In particular, using~\eqref{eq:nabla_tau}, 
\begin{equation} \label{eq:nabla_tau_spindown}
	 \dfrac{\bj}{\rho} - \dfrac{\Im (\overline{\sigma} {\nabla \sigma})}{\rho \rho^\spindown} = \dfrac{\bj - \rho^\spinup \nabla \tau}{\rho}  =  \nabla \tau^\spindown 
\end{equation}
is curl-free, and so is
\begin{equation} \label{eq:nabla_tau_spinup}
	\dfrac{\bj}{\rho} + \dfrac{\Im (\overline{\sigma} {\nabla \sigma})}{\rho \rho^\spinup} = \nabla \tau^\spinup.
\end{equation} 
%

%
%

\subsection{Proof of Theorem~\ref{th:CSDFT}}
\label{sec:proof_CSDFT}

We break the proof in several steps. \\

\noindent \textbf{Step 1: Any $R \in \cC_N$ can be written as $R = R_1 + R_2 + R_3$ with $R_k \in \cC_{N_k}^0$, $N_k \ge 4$.}\\

Let $R = \begin{pmatrix} \rho^\spinup & \sigma \\ \overline{\sigma} & \rho^\spindown \end{pmatrix} \in \cC_N$, with $N \ge 12$. We write $\sqrt{R} =  \begin{pmatrix} r^\spinup & s \\ \overline{s} & r^\spindown \end{pmatrix}$, with $r^{\spinup}, r^{\spindown} \in \left( H^1(\R^3) \right)^2$ and $s$ in $H^1(\R^3, \C)$. We write $R = R^\spinup + R^\spindown$ where $R^{\spinup / \spindown}$ were defined in~\eqref{eq:decomposition_convex}.
As in the proof of Theorem~\ref{th:SDFT} for the case $N=2$, $R^{\spinup / \spindown}$ are hermitian, of null determinant, and $\sqrt{R^{\spinup / \spindown}} \in \cM_{2 \times 2} \left( H^1(\R^3, \C) \right)$. However, it may hold that $\int \tr_{\C^2} [R^\spinup] \notin \N^*$, so that $R^\spinup$ is not in $\cC_M^0$ for some $M \in \N^*$. In order to handle this difficulty, we will distribute the mass of $R^\spinup$ and $R^\spindown$ into three density-matrices. \\
More specifically, let us suppose without loss of generality that $\int \tr_{\C^2} [R^\spinup] \ge \int \tr_{\C^2} [R^\spindown] $. We set
\begin{equation} \label{eq:decomposition123}
\begin{aligned}
	R_1 & = (1 - \xi_1) R^\spinup + \xi_2 R^\spindown \\
	R_2 & = \xi_1 (1 - \xi_3) R^\spinup \\
	R_3 & = (1 - \xi_2) R^\spindown + \xi_3 R^\spinup,
\end{aligned}
\end{equation}
where $\xi_1, \xi_2, \xi_3$ are suitable non-decreasing functions in $C^\infty(\R^3)$, that depends only on (say) $r_1$, such that, for $1 \le k \le 3$, $0 \le \xi_k \le 1$. We will choose them of the form $\xi_k(\br) = 0$ for $x_1 < \alpha_k$ and $\xi_k(\br) = 1$ for all $x_1 \ge \beta_k > \alpha_k$, and such that
\begin{equation} \label{eq:(1-xi)xi}
	(1 - \xi_1)\xi_2 = (1 - \xi_2) \xi_3 = (1 - \xi_1) \xi_3 = 0.
\end{equation}
Finally, these functions are tuned so that $\int \tr_{\C^2}(R_k) \in \N^*$ and $\int \tr_{\C^2}(R_k) \ge 4$ for all $1 \le k \le 3$ (see Figure~\ref{fig:cut_off} for an example of such a triplet $(\xi_1, \xi_2, \xi_3)$). Although it is not difficult to convince oneself that such functions $\xi_k$  exist, we provide a full proof of this fact in the Appendix.\\

From~\eqref{eq:(1-xi)xi}, it holds that for all $1 \le k \le 3$, $R_k \in \cC_{N_K}^0$, and that $R_1 + R_2 + R_3 = R^\spinup + R^\spindown = R$.

\begin{figure}[!h]
\begin{center}
\begin{tikzpicture}[scale=0.5]

	\draw[->] (-7, 4) -> (7, 4);
	\draw[black, line width=2] (-7, 6) -- (-4, 6);
	\draw[black, line width=2] (-4,6)	.. controls (-3,6) and (-3,4) .. (-2,4);
	
	\draw[gray, line width=2] (1, 6) -- (7, 6);
	\draw[gray, line width=2] (-1,4)	.. controls (0,4) and (0,6) .. (1,6);
	
	
	\node at (-6, 5) {$(1 - \xi_1)$};
	\node at (2.5, 5) {$\xi_2$};

	\draw (0, 7) -- (0, -5); 
	\draw (-3, 7) -- (-3, -5);
	\draw (4, 7) -- (4, -5);

	\draw[->] (-7, 0) -> (7, 0);
	\draw[black, line width=2] (-2, 2) -- (3, 2);
	\draw[black, line width=2] (-4,0)	.. controls (-3,0) and (-3,2) .. (-2,2);
	\draw[black, line width=2] (3,2)	.. controls (4,2) and (4,0) .. (5,0);
	
	\node at (0, 1) {$ \xi_1 (1 -  \xi_3)$};

	\draw[->] (-7, -4) -> (7, -4);
	\draw[gray, line width=2] (-7, -2) -- (-1, -2);
	
	\draw[gray, line width=2] (-1,-2)	.. controls (0,-2) and (0,-4) .. (1,-4);
	\draw[black, line width=2] (5, -2) -- (7, -2);
	\draw[black, line width=2] (3, -4) ..controls (4,-4) and (4, -2) .. (5, -2);

	\node at (-4.5, -3) {$(1 - \xi_2)$};
	\node at (6, -3) {$\xi_3$};
	
	\node at (-8.5, 5) {(a)};
	\node at (-8.5, 1) {(b)};
	\node at (-8.5, -3) {(c)};

\end{tikzpicture}
\caption{Weights of the matrices $R^\spinup$ (black) and $R^\spindown$ (gray) in (a) $R_1 = (1 - \xi_1)R^\spinup + \xi_2 R^\spindown$, (b) $R_2 = \xi_2 (1 - \xi_3) R^\spinup$ and $R_3 = (1 - \xi_2)R^\spinup + \xi_3 R^\spindown$.
}
\label{fig:cut_off}
\end{center}
\end{figure}

\noindent \textbf{Step 2 : The pair $(R_1, \bj_1)$ is representable by a Slater determinant.} \\
In order to simplify the notation, we introduce the total densities of $R^\spinup$ and $R^\spindown$:
\[
	f^\spinup := | r^\spinup |^2 + | s |^2 \quad \text{and} \quad f^\spindown := | r^\spindown |^2 + | s |^2.
\]
Recall that $\rho = f^\spinup + f^\spindown$. We consider the previous decomposition $R = R_1 + R_2 + R_3$, and we decompose $\bj$ in a similar fashion. More specifically, we write $\bj = \bj_1 + \bj_2 + \bj_3$ with
\begin{equation} \label{eq:j123}
\begin{aligned}
	\bj_1 & = (1 - \xi_1)  \left( \dfrac{f^\spinup}{\rho} \bj - \Im(\overline{s} \nabla s) \right) + \xi_2 \left( \dfrac{f^\spindown}{\rho} \bj + \Im(\overline{s} \nabla s) \right), \\
	\bj_2 & = \xi_1 (1 - \xi_3)  \left( \dfrac{f^\spinup}{\rho} \bj - \Im(\overline{s} \nabla s) \right), \\
	\bj_3 & = (1 - \xi_2) \left( \dfrac{f^\spindown}{\rho} \bj + \Im(\overline{s} \nabla s) \right) + \xi_3  \left( \dfrac{f^\spinup}{\rho} \bj - \Im(\overline{s} \nabla s) \right).
\end{aligned}
\end{equation}
Let us show that the pair $(R_1, \bj_1)$ is representable. Following~\cite{Lieb2013}, we introduce
\[
	\xi(x) = \dfrac{1}{m} \int_{-\infty}^x \dfrac{1}{(1 + y^2)^{(1 + \delta)/2}} \rd y,
\]
where $\delta$ is the one in~\eqref{eq:Liebcond}, and $m$ is a constant chosen such that $\xi(\infty) = 1$. We then introduce
\begin{equation} \label{eq:def_etak}
\begin{aligned}
	\eta_{1,1}(\br) & = \dfrac{2}{N} \xi(\br + \alpha), \\
	\eta_{1,2}(\br) & = \dfrac{2}{N-1} \xi( x_1 + \beta) (1 - \eta_1(\br)), \\
	\eta_{1,3}(\br) & = \dfrac{2}{N-2} \xi(x_2 + \gamma) (1 - \eta_1(\br) - \eta_2(\br)), \\
	\eta_{1,k}(\br) & = \dfrac{1}{N-3} (1 - \eta_1(\br) - \eta_2(\br) - \eta_3(\br)) \quad \text{for} \quad 4 \le k \le N,
\end{aligned}
\end{equation}
where $\alpha, \beta, \gamma$ are tuned so that, if $\rho_1 := \tr_{\C^2}R_1$ denotes the total density of $R_1$, 
\begin{equation} \label{eq:normalized_etaK}
	\forall 1 \le k \le N_k, \quad \int_{\R^3} \eta_{1,k} \rho_1 = 1.
\end{equation}
It can be checked (see~\cite{Lieb2013}) that $\eta_{1,k} \ge 0$ and $\sum_{k=1}^N \eta_{1,k} = 1$.
We seek orbitals of the form
\[
	\Phi_{1,k} := \sqrt{\eta_{1,k}} \left( \sqrt{(1 - \xi_1)} \begin{pmatrix} r^\spinup \\ \overline{s} \end{pmatrix} + \sqrt{\xi_2} \begin{pmatrix} s \\ r^\spindown \end{pmatrix} \right) \re^{\ri u_{1,k}}, \ 1 \le k \le N_1,
\]
and where the phases $u_{1,k}$ will be chosen carefully later. From~\eqref{eq:(1-xi)xi}, we recall that $(1 - \xi_1) \xi_2= 0$, so that, by construction, $\Phi_{1, k}$ is normalized, and
\[
	\Phi_{1, k} \Phi_{1, k}^* = \eta_{1,k} R_1.
\]
Let us suppose for now that the phases $u_{1,k}$ are chosen so that the orbitals are orthogonal. This will indeed be achieved thanks to the Lazarev-Lieb orthogonalization process (see Lemma~\ref{lem:HR}). Then, $\Psi_1 := \cS[\Phi_{1, 1}, \ldots, \Phi_{1, N}]$ indeed represents the spin-density matrix $R_1$. The paramagnetic current of $\Psi$ is (we recall that $r^\spinup$ and $r^\spindown$ are real-valued, and we write $s = | s | \re^{\ri \tau}$ for simplicity)
\begin{align*}
	\hspace{-0.5em} \bj_{\Psi} & = \sum_{k=1}^{N_1} \eta_{1,k} (1 - \xi_1) \left( | r^\spinup |^2 \nabla u_{1,k} + | s |^2 \nabla (-\tau + u_{1,k}) \right) +\\
		& \quad + \sum_{k=1}^{N_1} \eta_{1,k} \xi_2 \left( | s |^2 \nabla(\tau + u_{1,k}) + | r^\spindown |^2 \nabla u_{1,k} \right) \\ 
		 & = \left( (1 - \xi_1) f^\spinup + \xi_2 f^\spindown \right) \left( \sum_{k=1}^{N_1} \eta_{1,k} \nabla u_{1,k} \right) + 
		 \left( \xi_2 - (1 - \xi_1) \right)| s |^2 \nabla \tau.
\end{align*}
 Since $| s |^2 \nabla \tau = \Im (\overline{s} \nabla s)$, this current is equal to the target current $\bj_1$ defined in~\eqref{eq:j123} if and only if
\begin{equation} \label{eq:etaK_uK}
	\rho_1 \dfrac{\bj}{\rho} = \rho_1 \sum_{k=1}^{N_1} \eta_k \nabla u_{1,k}.
\end{equation}
In~\cite{Lieb2013}, Lieb and Schrader provided an explicit solution of this system when $N_1 \ge 4$. We do not repeat the proof, but emphasize on the fact that because condition~\eqref{eq:Liebcond} is satisfied by hypothesis, the phases $u_{1,k}$ can be chosen to be functions of $r_1$ only, and to have bounded derivatives. In particular, the functions $\Phi_{1, k}$ are in $H^1(\R^3, \C^2)$. Also, as their proof relies on the Lazarev-Lieb orthogonalization process, it is possible to choose the phases $u_{1,k}$ so that the functions $\Phi_{1, k}$ are orthogonal, and orthogonal to a finite-dimensional subspace of $L^2(\R^3, \C^2)$. \\

Altogether, we proved that the pair $(R_1, \bj_1)$ is representable by the Slater determinant $\sS[\Phi_{1, 1}, \ldots, \Phi_{1, N_1}]$. \\

\noindent \textbf{Step 3: Representability of $(R_2, \bj_2)$ and $(R_3, \bj_3)$, and finally of $(R, \bj)$.}\\
In order to represent the pair $(R_2, \bj_2)$, we first construct the functions $\eta_{2,k}$ for $1 \le k \le N_2$ of the form~\eqref{eq:def_etak} so that~\eqref{eq:normalized_etaK} holds for $\rho_2 := \tr_{\C^2} R_2$. We then seek orbitals of the form
\[
	\Phi_{2,k} := \sqrt{\eta_{2,k} \xi_1(1 - \xi_3)} \begin{pmatrix} r^\spinup \\ \overline{s} \end{pmatrix} \re^{\ri u_{2,k}}, \quad \text{for} \quad 1 \le k \le N_2.
\]
Reasoning as above, the Slater determinant of these orbitals represents the pair $(R_2, \bj_2)$ if and only if
\[
	\rho_2 \dfrac{\bj_2}{\rho} = \rho_2 \sum_{k=1}^{N_2} \eta_{2,k} \nabla u_{2,k}.
\]
Again, due to the fact that $N_2 \ge 4$, this equation admits a solution. Moreover, it is possible to choose the phases $u_{2,k}$ so that the functions $\Phi_{2,k}$ are orthogonal to the previously constructed $\Phi_{1,k}$. \\

We repeat again this argument for the pair $(R_3, \bj_3)$. Once the new set of functions $\eta_{3,k}$ is constructed, we seek orbitals of the form
\[
	\Phi_{3,k} := \sqrt{\eta_{3,k}} \left( \sqrt{(1 - \xi_2)} \begin{pmatrix} s \\ r^\spindown \end{pmatrix} + \sqrt{\xi_3} \begin{pmatrix} r^\spinup \\ \overline{s} \end{pmatrix} \right) \re^{\ri u_{3,k}}
\]
and construct the phases so that the functions $\Phi_{3,k}$ are orthogonal to the functions $\Phi_{1, k}$ and $\Phi_{2,k}$. \\

Altogether, the pair $(R, \bj)$ is represented by the (finite energy) Slater determinant $\sS[ \Phi_{1,1}, \ldots,  \Phi_{1, N_1}, \Phi_{2,1}, \ldots, \Phi_{2, N_2}, \Phi_{3,1}, \ldots, \Phi_{3, N_3}]$, which concludes the proof.

%
%
%

\section*{Acknowledgements}
I am very grateful to E. Cancès for his suggestions and his help.

%
%

\section{Appendix}

We explain in this section how to construct three functions $\xi_1, \xi_2, \xi_3 \in \left( C^\infty(\R)\right)^3$ like in Figure~\ref{fig:cut_off}. In order to simplify the notation, we introduce
\begin{align*}
	f(r) & := \iint_{\R \times \R} \tr_{\C^2} (R^\spindown) (r, r_2, r_3) \ \rd r_2 \rd r_3, \\ 
	g(r) & := \iint_{\R \times \R}  \tr_{\C^2} (R^\spinup) (r, y, z) \ \rd r_2 \rd r_3,
\end{align*}
where $R^\spinup, R^\spindown$ were defined in~\eqref{eq:decomposition_convex}. We denote by 
\[
	F(\alpha) = \int_{-\infty}^\alpha f(x) \rd x  \quad \text{and} \quad G(\alpha) = \int_{-\infty}^\alpha g(x) \rd x,
\]
and finally $\cF = F(\infty) = \int_\R f$ and $\cG = G(\infty) = \int_\R g$. Note that $F$ and $G$ are continuous non-decreasing functions going from $0$ to $\cF$ (respectively $\cG$), and that it holds $\cF + \cG = N$. Let us suppose without loss of generality that $\cF \le \cG$, so that $0 \le \cF \le N/2 \le \cG \le N$. If $\cF = 0$, then $R^\spindown = 0$ and we can choose $R_1 = R_2 = (4/N)R^\spinup \in \cC_4^0$ and $R_3 = (N-8)/N R^\spinup \in \cC_{N-8}^0$. Since $N \ge 12$, it holds $N - 8 \ge 4$, so that this is the desired decomposition. We now consider the case where $\cF \neq 0$. \\

In order to keep the notation simple, we will only study the case $\cF < 8$ (the case $\cF > 8$ is similar by replacing the integer $4$ by a greater integer $M$ such that $\cF < 2M < N-4$ in the sequel). We seek for $\alpha$ such that
\[
	\left\{	\begin{array}{l}
		\int_{-\infty}^\alpha f(x) \rd x < 4 \quad \text{and} \quad \int_{-\infty}^\alpha f(x) + \int_{\alpha}^\infty g(x) > 4, \\
		~\\
		 \int_{\alpha}^\infty f(x) \rd x < 4 \quad \text{and} \quad \int_{-\infty}^\alpha g(x) \rd x + \int_{\alpha}^\infty f(x) \rd x > 4,
	\end{array} \right.
\]
or equivalently
\begin{equation*}
	\left\{	\begin{array}{l}
		 F(\alpha) < 4 \quad \text{and} \quad F(\alpha) + \cG - G(\alpha) > 4, \\
		 \cF - F(\alpha) <  4 \quad \text{and} \quad \cF - F(\alpha) + G(\alpha) > 4,
	\end{array} \right.
\end{equation*}
that is
\begin{equation} \label{eq:system}
	\cF - 4 < F(\alpha) < 4,
	\quad \text{and} \quad
	F(\alpha) + 4 - \cF < G(\alpha) < F(\alpha) + \cG - 4.
\end{equation}

Let $\alpha_{(\cF-4)}$ be such that $F(\alpha_{(\cF-4)}) = \cF - 4$ (with $\alpha_{(\cF-4)} = - \infty$ if $\cF \le 4$), and $\alpha_{(4)}$ be such that $F(\alpha_{(4)}) = 4$ (with $\alpha_{(4)} = +\infty$ if $\cF \le 4$). As $F$ is continuous non-decreasing, the first equation of~\eqref{eq:system} is satisfied whenever $\alpha_{(\cF-4)} < \alpha < \alpha_{(4)}$.\\
The function $[\alpha_{(\cF-4)}, \alpha_4] \ni \alpha \mapsto m(\alpha) := F(\alpha) + 4 - \cF$ goes continuously and non-decreasingly from $0$ to $8 - \cF$, and the function $[\alpha_{(\cF-4)}, \alpha_4] \ni \alpha \mapsto M(\alpha) :=  F(\alpha) + \cG - 4$ goes continuously and non-decreasingly from $N - 8$ to $\cG$ between $\alpha_{(\cF-4)}$ and $\alpha_{(4)}$. In particular, since $G(\alpha)$ goes continuously and non-decreasingly from $0$ to $\cG$, only three cases may happen: \\
\begin{itemize}
	\item There exists $\alpha_0 \in ( \alpha_{(\cF-4)}, \alpha_{(4)})$ such that $m(\alpha_0) < G(\alpha_0) < M(\alpha_0)$. In this case, \eqref{eq:system} holds for $\alpha = \alpha_0$. By continuity, there exists $\varepsilon > 0$ such that
	\[
		\left\{ \begin{array}{l}
			F(\alpha + \varepsilon) < 4, \\
			F(\alpha) + \cG - G(\alpha + \varepsilon) > 4, \\
			G(\alpha) + \cF - F(\alpha + \varepsilon) > 4.
		\end{array} \right.
	\]
	Let $\xi_2 \in C^\infty(\R)$ be a non-decreasing function such that $\xi_2(x) = 0$ for $x < \alpha$ and $\xi_2(x) = 1$ for $x > \alpha + \varepsilon$. Then, as $0 \le \xi_2 \le 1$, it holds that:
	\[
		\int_{\R} (1 - \xi_2) f \le F(\alpha + \varepsilon) < 4 
	\]
	and
	\[
		\int_{\R} (1 - \xi_2) f + \int_{\alpha + \varepsilon}^\infty g \ge F(\alpha) + \cG - G(\alpha + \varepsilon) > 4.
	\]
	We deduce that there exists an non-decreasing function $\xi_3 \in C^\infty(\R)$ such that $\xi_3(x) = 0$ for $x < \alpha + \varepsilon$, and such that
	\[
		\int_{\R} (1 - \xi_2) f + \xi_3 g = 4.
	\]
	Note that $(1 - \xi_2) \xi_3 = 0$. On the other hand, from
	\[
		\left\{ \begin{array}{l}
			\int_{\R} \xi_2 f \le \cF - F(\alpha) < 4 \\
			~\\
			\int_{\R} \xi_2 f + \int_{-\infty}^\alpha g \ge \cF - F(\alpha + \varepsilon) + G(\alpha) > 4,
		\end{array} \right.
	\]
	we deduce that there exists an non-decreasing function $\xi_1 \in C^\infty(\R)$ such that $\xi_1(x) = 1$ for $x > \alpha$,
	\[
		\int_{\R} (1 - \xi_1) g + \xi_2 f = 4.
	\]
	and $(1 - \xi_1) \xi_2 = (1 - \xi_1) \xi_3 = 0$. Finally, we set
	\begin{align*} 
		R_1 & = (1 - \xi_1) R^\spinup + \xi_2 R^\spindown \\
		R_2 & = \xi_1 (1 - \xi_3) R^\spinup \\
		R_3 & = (1 - \xi_2) R^\spindown + \xi_3 R^\spinup.
	\end{align*}
	By construction, $R = R^\spinup + R^\spindown = R_1 + R_2 + R_3$, $R_1 \in \cC_4^0$ and $R_3 \in \cC_4^0$. We deduce that $R_4 \in \cC_{N-8}^0$. Together with the fact that $N \ge 12$, this leads to the desire decomposition.

	\item For all $\alpha \in ( \alpha_{(\cF-4)}, \alpha_{(4)})$, it holds $G(\alpha) < m(\alpha)$. Note that this may only happen if $m(\alpha_{(4)}) > 0$, or $\cF < 4$, so that $\cG > N - 4 \ge 8$. It holds $G(\alpha_{(\cF-4)}) = 0$, so that $g(r)$ is null for $r < \alpha_{(\cF-4)}$. Let $\alpha_0$ be such that $\alpha_{(\cF-4)} < \alpha_0 < \alpha_{(4)}$. As
	\[
		\int_{\R} f = \cF > 4 \quad \text{and} \quad \int_{\alpha_0}^\infty = \cF - F(\alpha_0) < 4,
	\]
	there exists a non-decreasing function $\xi_1 \in C^\infty(\R)$ satisfying $\xi_1(x) = 1$ for $x \ge \alpha_0$ and such that
	\[
		\int_{\R} \xi_1 f = 4.
	\]
	Now, since $G(\alpha_{(4)}) < m(\alpha_{(4)}) = 8 - \cF$, it holds that
	\[
		\left\{ \begin{array}{l}
		\int_{\R} (1 - \xi_1) f \le F(\alpha_{(4)}) = 4  \\
		~\\
		\int_{R} (1 - \xi_1) f + \int_{\alpha_0}^\infty g \ge F(\alpha_{(\cF-4)}) + \cG - G(\alpha_{(4)}) > 4. 
		\end{array}\right.
	\]
	There exists a non-decreasing function $\xi_2 \in C^\infty(\R)$ satisfying $\xi_2 (x) = 0$ for $x \le \alpha_0$ and such that
	\[
		\int_{\R} (1 - \xi_1) f + \xi_2 g = 4.
	\]
	Note that $(1 - \xi_1) \xi_2 = 0$. Finally, we set
	\begin{align*} 
		R_1 & = \xi_1 R^\spindown \\
		R_2 & = (1-\xi_2) R^\spinup \\
		R_3 & = \xi_2 R^\spinup + (1 - \xi_1) R^\spindown.
	\end{align*}
	By construction, it holds that $R = R_1 + R_2 + R_3$, and that $R_1 \in \cC_4^0$ and $R_3 \in \cC_4^0$. We deduce $R_2 \in \cC_{N-8}^0$, and the result follows.

	\item For all $\alpha \in ( \alpha_{(\cF-4)}, \alpha_{(4)})$, it holds $\cG(\alpha) > M(\alpha)$. This case is similar than the previous one.
\end{itemize}

\bibliographystyle{unsrt.bst}
\bibliography{pureState_v4_ARXIV}

\end{document}